\documentclass[%12pt,
a4paper]{article}
\usepackage{amsmath, amsfonts,
   amssymb, %amscd,
   amsthm, amsbsy, euscript}

\usepackage{diagrams}
 \diagramstyle[small,scriptlabels,
 midshaft,nohug]
 \newarrow{TeXto}----{->}
 
 \newcommand{\lto}{\lTeXto}

 \newcommand{\ruto}{\ruTeXto}
 
 \newcommand{\luto}{\luTeXto}
 \newcommand{\ldto}{\ldTeXto}
 \newarrow{AreEqual}=====

\newtheorem{theor}{Theorem}
\theoremstyle{definition}

\newtheorem{state}[theor]{Proposition}
\newtheorem{lemma}[theor]{Lemma}
\newtheorem{cor}[theor]{Corollary}

\newtheorem*{defNo}{Definition}

\newtheorem{example}{Example}%[section]

\theoremstyle{remark}
\newtheorem{rem}{Remark}
\newtheorem*{remNo}{Remark}

\newcommand{\cEv}{\partial}    %{\EuScript{E}}
\newcommand{\BBR}{\mathbb{R}}\newcommand{\BBC}{\mathbb{C}}

\newcommand{\cE}{\mathcal{E}}
\newcommand{\cEL}{\mathcal{E}_{\IL}}

\newcommand{\cEToda}{{\cE}_{\text{\textup{Toda}}}}

\newcommand{\cH}{\mathcal{H}}

\newcommand{\cL}{\mathcal{L}}

\newcommand{\cU}{\mathcal{U}}
\newcommand{\bu}{{\boldsymbol{u}}}
\newcommand{\bphi}{{\boldsymbol{\phi}}}
\newcommand{\bE}{\mathbf{E}}

\newcommand{\bun}{\mathbf{1}}

\newcommand{\gm}{\mathfrak{m}}
\newcommand{\gothg}{\mathfrak{g}}
\newcommand{\gA}{\mathfrak{A}}
\newcommand{\gB}{\mathfrak{B}}
\newcommand{\veps}{\varepsilon}
\newcommand{\vph}{\varphi}
\newcommand{\dd}{\partial}
\newcommand{\Id}{{\mathrm d}}

\newcommand{\IL}{{\mathrm L}}

\DeclareMathOperator{\img}{im}

\DeclareMathOperator{\sym}{sym}

\DeclareMathOperator{\arcsinh}{arcsinh}

\DeclareMathOperator{\CDiff}{\mathcal{C}Diff}

\DeclareMathOperator{\ord}{ord}

\DeclareMathOperator{\ad}{ad}
\DeclareMathOperator{\Sol}{Sol}

\newcommand{\by}[1]{\textit{{#1}}}
\newcommand{\jour}[1]{\textit{{#1}}}
\newcommand{\vol}[1]{\textbf{{#1}}}
\newcommand{\book}[1]{\textrm{{#1}}}

\newcommand{\ib}[3]{ \{\!\{ {#1},{#2} \}\!\}_{{#3}} }

\title%[Symmetry algebras of 2D~Toda\/-\/like chains]%
{Symmetry algebras of Lagrangian Liouville\/-\/type systems}

\date{February 26, 2009; in final form May 18, 2009} %completed June 19, 2009. 

\author%[A. V. Kiselev]
{Arthemy V. Kiselev,%${}^{*,\dagger}$, %}
\thanks{%\address{%
        \textit{Address}:
Mathematical Institute, University of Utrecht, P.O.Box 80.010, 3508 TA Utrecht, The Netherlands.
\textit{E-mails}: [\texttt{A.V.Kiselev},
\texttt{J.W.vandeLeur}]\texttt{\symbol{"40}uu.nl}%
}
\quad %
Johan W. van de Leur${}^{*}$}
\begin{document}                                      %Dedicated to M.N.
\maketitle

\begin{abstract}
The generators and %all the 
commutation relations are calculated explicitly for higher symmetry algebras of a class of hyperbolic Euler\/--\/Lagrange systems of Liouville type (in particular, for 2D~Toda chains associated with semi\/-\/simple complex Lie algebras).

\noindent%
\end{abstract}

\paragraph*{Introduction.}
We give a %complete 
description of the generators and relations in
higher symmetry algebras for a class
of Darboux\/-\/integrable hyperbolic Euler\/--\/Lagrange systems
of Liouville type \cite{LeznovSmirnovShabat,ShabatYamilov,SokolovUMN}. 
There exist many non\/-\/equivalent definitions of this
type of PDEs \cite{LeznovSmirnovShabat,SokolovUMN,GurievaZhiber}; 
we investigate the systems $\cEL$ that admit as many first
integrals of the characteristic equations $D_y(w)\doteq0$ %on $\cEL$ 
and $D_x(\bar{w})\doteq0$ on $\cEL$ as there are unknown functions.
The 2D~Toda chains $\bu_{xy}=\exp(K\bu)$
associated with semi\/-\/simple complex Lie algebras are the most well studied example of such equations \cite{LeznovSmirnovShabat,ShabatYamilov,Leznov,Meshkov198x,Shabat95}. 
The systems of this class
are known to possess higher symmetries $\vph=\square({\bphi})$
that depend on free functional parameters 
$\smash{\bphi}=
{}^t\bigl(\phi_1(x,[w])$,\ $\ldots$,\ $\phi_r(x,[w])\bigr)$ and
belong to the image of matrix total differential operators $\square$ 
(linear operators in total derivatives)
\cite{Meshkov198x,Demskoi2004,TMPhGallipoli,SokStar}.
The existence of such operators $\square$ for Liouville\/-\/type systems was
observed in \cite{LeznovSmirnovShabat,Meshkov198x} 
and %explained in %recognized, understood
\cite{SokolovUMN,SokStar}, where the importance of the linearizations
$\smash{\ell_w^{(u)}}$ of the first integrals $w$ in the construction of $\square$
was revealed. In the paper \cite{TMPhGallipoli} we proved that the
additional assumption for $\cEL$ be Euler\/--\/Lagrange strengthens known
results and even makes the description of $\square$ explicit, see
formula \eqref{Square} below.
 
%We show that this property follows from the initial assumptions 
%on the class of equations $\cEL$. 
In this paper 
we establish the transformation rules for the %factoring 
operators $\square$ under unrelated reparametrizations
of the coordinates in their domains and images. 
We show that, under natural assumptions on the geometry of~$\cEL$,
the images of these operators are closed with respect to the commutation, whence the Lie algebra structure on their domains appears. 
We calculate the brackets on the domains explicitly, which yields, by
the push forward of the Lie algebra structure, %all 
the commutation relations in the symmetry algebras~$\sym\cEL$.
To do this, we introduce auxiliary Hamiltonian operators which have the same domain as~$\square$. %, which allows to reduce the problem.

\begin{remNo}
We do not assume the presence of a symmetry $x\leftrightarrow y$ in $\cEL$.
We work with `the $x$-\/half' of the algebra $\sym\cEL$ related to the first integrals
$w^i\in\ker D_y{\bigr|}_{\cEL}$; the reasonings hold for the respective 
`$y$-\/half' of $\sym\cEL$, and the two subalgebras commute between each other.
For the Euler\/--\/Lagrange systems $\cEL$ at hand, the integrals
$\bar{w}^{\bar{\imath}}\in\ker D_x{\bigr|}_{\cEL}$ are not used in the proofs,
unlike in \cite{SokStar} for arbitrary Liouville\/-\/type systems. 

The full list of assumptions on the systems~$\cEL$ and their integrals
is given in our main Theorem~\ref{IspHO}, see also Remark~\ref{RemProofX} 
on p.~\pageref{RemProofX}.    %section~2. %\ref{SecCommut}
 %Second, the  in section~1 %\ref{SecSym} are less restrictive than in 
 %FormulaForSquare}.
 %some of our requirements are excessive and are made for the sake of
 %transparency only. Namely, 
However, the reasonings in section~1 %\ref{SecSym}
hold under less restrictive conditions. In particular, %even if 
the number % $r$ 
of first integrals $w^1$,\ $\ldots$,\ $w^r$
for the characteristic equation on $\cEL$ can be less than the number % $m$ 
of the unknowns $u^1$,\ $\ldots$,\ $u^m$
in $\cEL$. In that case, the auxiliary $(r\times r)$-\/matrix operators $\smash{\hat{A}_k}$ defined in \eqref{Quattro} become smaller in size
but remain Hamiltonian (see \cite{TMPhGallipoli} for the second Poisson structure for KdV provided by the 2D~Toda chains with a unique
integral). 
%The proof of Corollary \ref{SymStructure} below holds for the nontrivial
%topology of the bundles with fibre coordinates $u^i$, base points $(x,y)$,
%and the projection $\pi$ (e.g., if $(x,y)$ is a point of the 
%torus $\mathbb{T}^2$).
\end{remNo}

The paper is organized as follows. First we define the operators $\square$
that determine symmetry generators for the systems $\cEL$ and
introduce auxiliary Hamiltonian operators. Here %, in particular, 
we re\/-\/derive the higher Poisson structures for the Drinfel'd\/--\/Sokolov 
hierarchies \cite{DSViniti84} on 2D~Toda chains related to semi\/-\/simple
complex Lie algebras; an example 
is given for the $\mathsf{A}_2$-\/Toda chain. %In sect. \ref{SecCommut} Then 
Then in section~2 %\ref{SecCommut} 
we establish the commutation closure for images of the operators~$\square$ and
calculate the structural %commutation 
relations in the %symmetry 
algebras $\sym\cEL$; 
an illustration is given for the Kaup\/--\/Boussinesq equation.
Finally, in section~3 %\ref{SecNonEL}
we discuss some properties of the operators that yield symmetries
of non\/-\/Lagrangian Liouville\/-\/type systems.

All notions and constructions from geometry of PDE are standard \cite{ClassSym,
Olver,Opava}. We follow the notation of \cite{TMPhGallipoli,TMPhGardner,Protaras}. %completes a part of results in 
This paper develops further the concept of~\cite{TMPhGallipoli}.

\paragraph{1. Symmetry generators for $\cEL$.}\indent\par\noindent\label{SecSym}
\begin{defNo}%\label{DefLiouType}
A \emph{Liouville\/-\/type system}
$\cE$ is a system $\{\bu_{xy}=f(\bu,\bu_x,\bu_y;x,y)\}$ of $m$ 
hyperbolic equations 
upon $\bu=(u^1,\ldots,u^m)$ which admits nontrivial \emph{first integrals} 
\[
w^1,\ \ldots,\ w^r\in\ker D_y{\bigr|}_{\cE};\quad 
\bar{w}^1,\ \ldots,\ \bar{w}^{\bar{r}}\in\ker D_x{\bigr|}_{\cE},\qquad
0<r,\bar{r}\leq m,
\]   %\in C^\infty(\cE) 
for the linear first order characteristic equations
   %such that the relations
${D_y\bigr|}_{\cE}(w^i)\doteq0$ 
and ${D_x\bigr|}_{\cE}(\bar{w}^{\bar{\jmath}})\doteq0$ 
that hold by virtue ($\doteq$) of $\cE$.
\end{defNo}     
 %and such that all conservation laws for $\cE$ belong to %are of the form
 %$\int I(x,[w])\,\Id x\oplus\int \bar{I}(y,[\bar{w}])\,\Id y$
 %with differential functions $I$ and $\bar{I}$.

\begin{example}\label{ExTodaAreLiou}
In \cite{ShabatYamilov} it was proved that the 2D~Toda chains \cite{Leznov}
$u^i_{xy}=\exp(K^i_{\,j}u^j)$ related to semi\/-\/simple complex Lie
algebras with the Cartan matrices $K$ admit maximal ($r=\bar{r}=m$) sets of the
integrals. Various methods for reconstruction of $w^i$, 
$\bar{w}^{\bar{\jmath}}$ for these
exponential\/-\/nonlinear Toda chains were proposed 
in \cite{SokolovUMN,GurievaZhiber,Shabat95}. 
   %Using the results of \cite{Shabat95}, we prove that 
The differential orders (after a shift by $-1$) of the integrals $w^1$, 
$\ldots$, $w^r$ w.r.t.\ $\bu$ are equal to the exponents of the corresponding 
semi\/-\/simple complex Lie algebras of rank $r$, 
which follows from \cite[p.\,21]{ShabatYamilov}.

For instance, in the sequel we
consider the Euler\/--\/Lagrange 2D~Toda system $\cEToda$
associated with the simple Lie
algebra $\mathfrak{sl}_3(\BBC)$, 
see \cite{LeznovSmirnovShabat,Leznov,Shabat95},
\begin{equation}\label{A2Toda}
\cEToda=\bigl\{ u_{xy}=\exp(2u-v),\ v_{xy}=\exp(-u+2v),\qquad
K=\left(\begin{smallmatrix}\phantom{+}2&-1\\
 -1&\phantom{+}2\end{smallmatrix}\right)\bigr\}.
\end{equation}
The %minimal 
integrals of respective orders $2$ and $3$ for system \eqref{A2Toda} are
(e.g., see~\cite{Ibragimov})
%\begin{subequations}\label{IntA2}
%\begin{align*}  %\cite{LeznovSmirnovShabat,Shabat}
$w^1 %&
 =u_{xx}+v_{xx}-u_x^2+u_xv_x-v_x^2$ and %,\\ %\label{IntA2First}\\
$w^2 %&
=u_{xxx}-2u_xu_{xx}+u_xv_{xx}+u_x^2v_x-u_xv_x^2$.
%\end{align*}
%\end{subequations}
\end{example}

The generators $\vph=\square\bigl(\smash{{\bphi}}(x,[w])\bigr)$ of higher
symmetry algebras for Liouville\/-\/type equations are given by 
matrix total differential operators $\square$, see %in total derivatives 
\cite{SokolovUMN,Meshkov198x}. %LeznovSmirnovShabat
%For the Euler\/--\/Lagrange Liouville\/-\/type systems, the operators provide
%Noether symmetries $\vph_{\cL}$ with $\vec{\phi}=\delta\cH(x,[w])/\delta w$,
%whence the description of all symmetries follows.
For Euler\/--\/Lagrange Liouville\/-\/type systems
$\cEL=\{F\equiv\bE(\cL)=0\}$, 
see \cite{Demskoi2004,TMPhGallipoli,Startsev2006},
the existence of \emph{certain}
factorizations % \eqref{SymForLiou} 
for at least a \emph{part} of symmetries
is rigorous and can be readily seen as follows.
For integrals $w$ such that $D_y(w)=\nabla(F)$ vanishes on the differential
ideal $\{F=0\}^\infty$ by virtue of an operator $\nabla$, 
and for any $I(x,[w])$,
the generating section $\psi_I=\Bigl[\nabla^*\circ\bigl(\ell_w^{(u)}\bigr)^*
\circ\bigl(\ell_I^{(w)}\bigr)^*\Bigr](1)$ for a conservation law
$\int I\,\Id x$ solves the equations $\ell_{\bE(\cL)}^*(\psi_I)\doteq0$
on $\cEL$, see \cite{ClassSym,Olver,Opava}. 
The Helmholtz condition $\ell_{\bE(\cL)}=\ell^*_{\bE(\cL)}$ 
for the linearization (the Frech\'et derivative) implies that the vector
\begin{equation}\label{SquareForEL}
\vph[\bphi]=\Bigl[\nabla^*\circ\bigl(\ell_w^{(u)}\bigr)^*\Bigr]\bigl(
\bphi(x,[w])\bigr)\in\ker\ell_{\bE(\cL)}\bigr|_{\cEL}
\end{equation}
is a symmetry
of $\cEL$ for any $\bphi=\bigl(\ell_I^{(w)}\bigr)^*(1)=\bE_w(I\,\Id x)$.
A standard homological reasoning 
(see \cite[Ch. 5]{Olver} %Ex.5.32* a)-c)
or \cite[\S7.8]{Opava})
   %Lemma \ref{LBaseEuler} or Corollary \ref{SymStructure} below) 
shows that formula \eqref{SquareForEL}
yields symmetries of the %Euler\/--\/Lagrange
sys\-tem $\cEL$ even if sections $\bphi$
do not belong to the image of the variational derivative $\bE_w$ w.r.t.~$w$.

In this section we recall the construction of %well\/-\/defined 
operators $\square$ that determine %all (up to $x\leftrightarrow y$) 
symmetries for a class of Euler\/--\/Lagrange Liouville\/-\/type systems.
%The images of such operators are closed under the commutation whenever 
We suppose that the integrals $w$ are \emph{minimal}, meaning %that
$I\in\ker D_y\bigr|_{\cEL}$ implies $I=I(x,[w])$.

\begin{state}[\cite{TMPhGallipoli}]\label{NoetherSymTh}
Let $\kappa$ be an invertible %nondegenerate 
symmetric constant real $(m\times
m)$-\/matrix. Suppose that $\cL=\int L\,\Id x\Id y$ with the density
$L=-\tfrac{1}{2}\sum_{i,j}\kappa_{ij}u^i_xu^j_y-H_{\IL}(u;x,y)$ is the
Lagrangian of a Liouville\/-\/type equation $\cEL=\{\bE(\cL)=0\}$.
Let $\gm=\dd L/\dd u_y$ be the momenta % \cite{Dirac}
and $w(\gm)=(w^1,\ldots,w^r)$ be the minimal set of integrals for $\cEL$
that belong to the kernel of ${D_y\bigr|}_{\cEL}$.
Then the adjoint linearization
\begin{equation}\label{Square}
\square=\bigl(\ell_w^{(\gm)}\bigr)^*
\end{equation}
of the integrals w.r.t.\ the momenta yields %factors all 
Noether symmetries $\vph_{\cL}$ of $\cEL$: %which are given by
\begin{equation}\label{NoetherSym}
\vph_\cL = \square\bigl({\delta\cH}/{\delta w}\bigr)
\qquad\smash{\text{for any $\cH=\int H(x,[w])\,\Id x$.}}
\end{equation}
\end{state}

\begin{cor}\label{SymStructure}
Under the assumptions and notation of Proposition~\textup{\ref{NoetherSymTh},}
the section %$m$-\/tuple
\begin{equation}\label{SymForLiou}
\vph=\square\bigl({\bphi}(x,[w])\bigr) %\in\varkappa(\pi)
\end{equation}
is a %\textup{(}higher\textup{)}
symmetry of the Liouville\/-\/type equation $\cEL$ for any $r$-\/tuple
${\bphi}={}^t(\phi_1,\ldots,\phi_r)$. %\in\gf$.
\end{cor}

\begin{proof}%\label{ProofBase}
%The proof is standard and analogous to the one for Lemma \ref{LBaseEuler} 
%with the only alteration in the jet space at hand.
  %a reformulation of Lemma \ref{LBaseEuler} is proved on
  %p. \pageref{ProofBase}.
Consider the jet bundle $J^\infty(\xi)$ over the fibre bundle
$%\begin{equation}\label{KdVBundle}
\xi\colon\BBR^r\times\BBR\to\BBR
$ %\end{equation}
with the base $\BBR\ni x$
and the fibres $\BBR^r$ with coordinates $w^1,\ldots,w^r$.
By Proposition~\ref{NoetherSymTh}, the statement %of the theorem 
is valid for any $\smash{{\bphi}}$ %\in\gf=\hat{\varkappa}(\xi)$
  %=\widehat{\Gamma}(\mu_\infty^*(\mu))$
in the image of the variational derivative $\bE_w$. %, where $w=w[\gm]$.
Obviously, its image contains all variational covectors %$r$\/-\/tuples 
$\smash{{\bphi}}$ whose components $\phi_i(x)\in
C^\infty(\BBR)$ are functions on the base of the new %jet 
bundle $\xi$. The prolongation of the 
substitution $w=w\bigl[\gm[u]\bigr]\colon J^\infty(\pi)\to\Gamma(\xi)$ 
       %\colon\gf\to\varkappa(\pi)$
converts the components of sections $\smash{{\bphi}}$ %\in\gf$
to %elements of $\Gamma\bigl(\pi_\infty^*(\xi)\bigr)$.
smooth differential functions in $u$ (which denotes the set of fibre
coordinates in the bundle $\pi$ over the same base $\BBR\ni x$).
Now recall that $\square$ is an operator in total
derivatives $D_x$ % \eqref{CartanConnection}, 
whose action on 
  %$r$-\/tuples of
differential functions $f[u]$ %on the jet space $J^\infty(\pi)$
is defined by % \eqref{DefDxByJs} 
$j_\infty(s)\bigl(D_x(f)\bigr)\mathrel{{:}{=}}\frac{\dd}{\dd x}\bigl(
j_\infty(s)(f)\bigr)$
through the restrictions $j_\infty(s)(f)$ of $f$
onto the jets $j_\infty(s)$ of sections $u=s(x)$. %\smash{\vec{\phi}}(x)$, 
Hence we obtain $\phi_i(x)=\phi_i\bigl(x,\bigl[w[\gm[s(x)]]\bigr]\bigr)$,
and the assertion follows.
\end{proof}

\begin{remNo}
This proof combines taking the variational derivatives with respect to $w$
on one jet space with calculating the total derivatives of differential
functions in $u$ on the other jet space over the same base.
Whenever the two bundles coincide, the reasoning amounts to the definition
of $D_x$. Then it is called the \emph{substitution principle}
(\cite{Olver}, see a detailed discussion in \cite{Opava}).
\end{remNo}

\begin{theor}\label{Kind2}
Under differential reparametrizations $\tilde{w}=\tilde{w}[w]$ and
$\tilde{u}=\tilde{u}[u]$ of the coordinates $w^1,\ldots,w^r$ 
and $u^1,\ldots,u^m$ in the infinite jet bundles over $\xi$ and $\pi$ that
specify its domain and image, respectively, 
the operator $\square$ is transformed according to the formula
\begin{equation}\label{FrobAK}   %is Frobenius of second kind. 
\square\mapsto\tilde{\square}=\ell_{\tilde{u}}^{(u)}\circ \square\circ
 \bigl(\ell_{\tilde{w}}^{(w)}\bigr)^*
  {\Bigr|_{\substack{w=w[u]\\u=u[\tilde{u}]}}}.
\end{equation}
\end{theor}

\begin{proof}
The transformation $\tilde{\vph}=\ell_{\tilde{u}}^{(u)}(\vph)$ of the
velocities is obvious.
Under differential reparametrizations $w=w[\tilde{w}]$ of the integrals,
the sections $\bphi=\delta\cH/\delta w$ are transformed by
$\bphi=\bigl(\ell_{\tilde{w}}^{(w)}\bigr)^*(\tilde{\bphi})$,
thence $\square$ becomes well defined on $\img\bE_w$. Namely, it maps
variational covectors for the fibre bundle~$\xi$ to evolutionary derivations 
in the jet space over the other fibre bundle~$\pi$. Repeating the reasoning
used in the proof of Corollary \ref{SymStructure}, we establish the
transformation rule~\eqref{FrobAK} %and the commutation closure of 
for~$\square$ on the entire domain.
\end{proof}
 %Frobenius operator of second kind (\ref{Frob2}--\ref{FrobAK}).

In Theorem~\ref{Kind2} we showed that sections in the domain of 
the operator~$\square$ are transformed by the same rule as the arguments of
Hamiltonian operators. There is indeed a deep reason for that.

The integrals $w[\gm]$ of Euler\/--\/Lagrange Liouville\/-\/type
systems $\cEL$ determine the Miura sub\-sti\-tu\-ti\-ons from commutative
modified KdV\/-\/type Hamiltonian hierarchies $\gB$ of Noe\-ther symmetries
for $\cEL$ to completely integrable KdV\/-\/type hierarchies $\gA$ of
higher symmetries of the multi\/-\/component wave
equations $\cE_\varnothing=\{s_{xy}=0\}$, see below. 
A natural example is given by the potential modified KdV equation
$u_t=-\tfrac{1}{2}u_{xxx}+u_x^3%=\square(w)
$, which is transformed to the KdV equation
$w_t=-\tfrac{1}{2}w_{xxx}+3ww_x$ by $w=u_x^2-u_{xx}$.
This example was analysed %discussed 
in detail in \cite{TMPhGallipoli}.
The method for generating relevant differential substitutions by 
the integrals $w$ of Liouville\/-\/type systems was discovered 
in \cite{SokolovSubst}, %by V. V. Sokolov,
see \cite{SokolovUMN} for discussion. % and futher conjectures.
This fact was used in \cite{StartsevClass} for a classification of
the first\/-\/order differential substitutions.

The hierarchies $\gA$ and $\gB$ share the Hamiltonians
$\cH_i[\gm]=\cH_i\bigl[w[\gm]\bigr]$ through the Miura substitution $w[\gm]$.
The Hamiltonian structures for the Magri schemes of $\gA$ and $\gB$ are
correlated by %such that 
the operators $\square$, which map cosymmetries $\phi_i$
for the hierarchy $\gA$ to symmetries $\vph_{i+1}$ of the
modified hierarchy $\gB$. We stress that the first Hamiltonian operator
$\hat{B}_1=\bigl(\ell_\gm^{(u)}\bigr)^*$ for $\gB$ originates from the
differential constraint $\gm=\dd L/\dd u_y$ upon the coordinates $u$
and the momenta $\gm$ for $\cEL$.
Using the explicit form of the `junior' %primary 
operator $\smash{\hat{B}_1}$
and the differential\/-\/functional closure in $[w]$ for the velocities
of the integrals, see \cite{SokolovUMN}, we realize the classical scheme
for generating higher Poisson structures via Miura's 
substitutions \cite{KuperWilson}.

\begin{lemma}\label{AkIsHam}
Introduce the linear differential operator
\begin{equation}\label{Quattro}
\hat{A}_k=\square^*\circ\hat{B}_1\circ\square
\end{equation}
that maps variational covectors for the jet bundle $J^\infty(\xi)$ over $\xi$
to evolutionary vector fields on it, % $J^\infty(\xi)$,
$\smash{\hat{A}_k}\colon %\bigl(\ell_\gm^u\bigr)^*.
\Gamma\bigl(\widehat{\xi\,}\bigr)\mathbin{{\otimes}_{C^\infty(\BBR)}}
 C^\infty\bigl(J^\infty(\xi)\bigr) \to
\Gamma(\xi)\mathbin{{\otimes}_{C^\infty(\BBR)}}
 C^\infty\bigl(J^\infty(\xi)\bigr)$.
%which is completely determined by the Euler\/--\/Lagrange
%Liouville\/-\/type system $\cE_\IL$.
  %Note that all factors in the r.h.s.\ of \eqref{Quattro} are
  %specified explicitly
The operator \eqref{Quattro} is Hamiltonian and determines\footnote{%
In most cases, this is one of the \emph{higher} structures for $\gA$,
which is indicated by the subscript $k=k(\square,\gm)\geq2$. %usually, $k=2$.
The choice of the `junior' operator $\smash{\hat{A}_1}$ for $\gA$
is discussed in what follows.}
a Poisson structure for the KdV\/-\/type hierarchy $\gA$.
The coefficients of $\smash{\hat{A}_k}$ are differential functions in $w$.
\end{lemma}

\begin{proof}
By construction, the Poisson bracket % \eqref{PoissonEquiv} 
$\{\cH_1,\cH_2\}_{\hat{A}_k}=\bigl\langle\bE_w\cH_1,\hat{A}_k\bigl(\bE_w\cH_2\bigr)\bigr\rangle$
%given by $\hat{A}_k$ 
satisfies the equality
\begin{equation}\label{RecalcWviaU}
\bigl\{\cH_1[w],\cH_2[w]\bigr\}_{\hat{A}_k}=
   \bigl\{\cH_1[w[\gm]],\cH_2[w[\gm]]\bigr\}_{\hat{B}_1}.
\end{equation}
Therefore the left\/-\/hand side of \eqref{RecalcWviaU} is
bi\/-\/linear, skew\/-\/symmetric, and satisfies the Jacobi identity.
Fourth, it measures the velocity of the integrals $w$ along a Noether
symmetry of $\cEL$. Since evolutionary derivations are permutable with
the total derivative $D_y$, the velocity $\{\cH_1,\cH_2\}_{\hat{A}_k}$
lies in $\ker D_y\bigr|_{\cEL}$ and hence its density is a differential function of the minimal integrals $w$.
\end{proof}

The multi\/-\/component wave equation $\cE_\varnothing=\{s_{xy}=0\}$,
whose symmetries contain the hierarchy $\gA$ and such that %$s_x=w$
$\hat{A}_1=\bigl(\ell_w^{(s)}\bigr)^*$ encodes the differential constraint
between the coordinates $s$ and momenta $w$ for $\cE_\varnothing$,
 %in %potentiates the image of the Miura substitution, 
is chosen such that
the first structure $A_1=\smash{\hat{A}^{-1}_1}$ for $\gA$
 %, which is also 
 %determined by the constraint between the coordinates $s$ and 
 %the momenta $w$ , 
factors the higher Hamiltonian
structure for $\gB$. Hence %We thus have
$%\[
B_{k'}=\square\circ A_1\circ\square^*$, where
$k'=k'\bigl(\square,\bigl(\ell_w^{(s)}\bigr)^*\bigr)\geq2$.
%\]

\begin{example}\label{ExA2}
Consider the Euler\/--\/Lagrange 2D~Toda system \eqref{A2Toda}.
The density~$L$ of its Lagrangian is
\[
L=-\tfrac{1}{2}\bigl((2u_x-v_x)\cdot u_y + (2v_x-u_x)\cdot v_y\bigr)
  -\exp(2u-v)-\exp(2v-u).
\]
Therefore we introduce the momenta
$%\[
\gm^1\mathrel{{:}{=}}2u_x-v_x$ and $\gm^2\mathrel{{:}{=}}2v_x-u_x$,
%\]
whence we express the integrals as $w=w[\gm]$. %follows,
%\begin{align*}
%w^1 &= 3\gm^1_x+3\gm^2_x-(\gm^1)^2-\gm^1\gm^2-(\gm^2)^2,\\
%w^2 &= 2\gm^1_{xx}+\gm^2_{xx}-2\gm^1\gm^1_x-\gm^2\gm^1_x+\tfrac{2}{9}(\gm^1)^3
   %\\  &\quad 
%+\tfrac{1}{3}(\gm^1)^2\gm^2-\tfrac{1}{3}\gm^1(\gm^2)^2-\tfrac{2}{9}(\gm^2)^3.
%\end{align*}
All symmetries (%that depend on free differential\/-\/functional parameters,
up to $x\leftrightarrow y$) of \eqref{A2Toda} are of the form
$\vph=\square\bigl({\bphi}\bigl(x,[w^1],[w^2]\bigr)\bigr)$, where
${\bphi}={}^t(\phi_1,\phi_2)$ is a pair of arbitrary functions %\in\hat{\varkappa}(\xi)$
and the $(2\times2)$-\/matrix total differential operator 
  %in total derivatives 
is %given by formula \eqref{ISquare},
\[%\begin{equation}\label{SquareA2}
\square=\ell_{w^1,w^2}^{(\gm^1,\gm^2)}=
\begin{pmatrix}
u_x+D_x & 
 -\tfrac{2}{3}D_x^2-u_xD_x-\tfrac{1}{3}u_x^2-\tfrac{2}{3}u_xv_x
   +\tfrac{2}{3}v_x^2+\tfrac{1}{3}u_{xx}-\tfrac{2}{3}v_{xx} \\
v_x+D_x & 
 -\tfrac{1}{3}D_x^2+\tfrac{2}{3}u_{xx}-\tfrac{1}{3}v_{xx}
   -\tfrac{2}{3}u_x^2+\tfrac{2}{3}u_xv_x+\tfrac{1}{3}v_x^2 \end{pmatrix}.
\]%\end{equation}
%Its first column is contained in the encyclopaedia \cite{Ibragimov};
%an operator with both columns of order two, hence generating a linear subspace
%of infinite codimension 
%in the symmetry algebra for \eqref{A2Toda}, is derived in \cite{SokStar}.
The entries of the arising Hamiltonian operator %~\eqref{Quattro},
$\smash{\hat{A}_2}=\left\|A_{ij},
  %\begin{smallmatrix} A_{11} & A_{12}\\ A_{21} & A_{22}\end{smallmatrix}
1\leq i,j\leq2\right\|$ %\]
are~\cite{Olver}  %[Example 7.28, p.556 Rus.Ed.]
%where for the root system~$\mathsf{A}_2$ we have that
\begin{align*}
A_{11} &= 2 D_x^3 +2 w^1 D_x + w^1_x,\\
A_{12} &= -D_x^4 -w^1 D_x^2 + (3 w^2-2 w^1_x)\cdot D_x +(2 w^2_x-w^1_{xx}),\\
A_{21} &= D_x^4 + w^1 D_x^2 +3 w^2 D_x + w^2_x\\
A_{22} &= -\tfrac{2}{3} D_x^5 -\tfrac{4}{3} w^1 D_x^3 -2 w^1_x D_x^2 
 + (2 w^2_x-2 w^1_{xx}-\tfrac{2}{3} (w^1)^2)\cdot D_x \\
 {}&\qquad{}+\tfrac{1}{3}(3 w^2_{xx}-2 w^1_{xxx}-2 w^1 w^1_x).
\end{align*}
The shift $w^2\mapsto w^2+\lambda$ of the second integral
 %, and taking the velocity of the operator~$\smash{\hat{A}_k}$,
yields the `junior' Hamiltonian operator\footnote{The analogous operator 
$\hat{A}_1^{(1)}=\tfrac{d}{d\mu}{\bigr|}_{\mu=0}\bigl(\hat{A}_k\bigr)$, where
$w^1\mapsto w^1+\mu$, is not Hamiltonian.}
$\hat{A}_1^{(2)}
\mathrel{{:}{=}}\tfrac{d}{d\lambda}{\Bigr|}_{\lambda=0}
   \bigl(\hat{A}_2\bigr)
 =\left(\begin{smallmatrix}0 & 3D_x\\ 3D_x & 0\end{smallmatrix}
\right)$,
%\]
which is compatible with~$\smash{\hat{A}_2}$.
   %we obtain $\hat{A}_1^{(2)}$ that  the former.

%It is well known that the Boussinesq equation with a (removable)
%dissipation appears for the root system $\mathsf{A}_2$.
The pair $(\hat{A}_1^{(2)},\hat{A}_2)$ is the %well\/-\/known 
bi\/-\/Hamiltonian structure for the Boussinesq equation
%\begin{subequations}\label{Kiselev:EBous}
%\begin{align}
$w^1_t=2w^2_x-w^1_{xx}$,
$w^2_t=-\tfrac{2}{3}w^1_{xxx}-\tfrac{2}{3}w^1w^1_x+w^2_{xx}$.
%\end{align}
%\end{subequations}
The symmetry ${w}_x=\bigl(\hat{A}_2\circ{\delta}/{\delta{w}}\bigr)
\bigl(\int w^1\,{\mathrm{d}} x\bigr)$
starts the second sequence of Hamiltonian flows in 
the Boussinesq hierarchy~$\mathfrak{A}$.
The modified Boussinesq hierarchy~$\mathfrak{B}$ shares the two sequences of Hamiltonians with~$\gA$ by the Miura substitution $w=w[\mathfrak{m}]$ with $\mathfrak{m}=\mathfrak{m}[u]$. 
Namely, for any Hamiltonian~$\mathcal{H}[w]$, the flows
%\[
$u_\tau={\delta\mathcal{H}[\mathfrak{m}]}\bigr/{\delta\mathfrak{m}}$,
$\mathfrak{m}_\tau=-{\delta\mathcal{H}\bigl[\mathfrak{m}[u]\bigr]}\bigr/{\delta u}$
%\]
belong to% the modified hierarchy
~$\mathfrak{B}$.
The velocities~$u_\tau$ constitute the commutative subalgebra of Noether symmetries for the 2D~Toda chain~\eqref{A2Toda}.
\end{example}

\paragraph{2. Commutation relations in $\sym\cEL$.}\label{SecCommut}%\noindent%
In this section we prove the commutation closure for the images of 
operators~\eqref{Square} by using the well\/-\/known analogous property of
the auxiliary Hamiltonian operators~\eqref{Quattro}. At once, we describe %all 
the %structural %commutation 
relations in the symmetry algebra generated by~\eqref{SymForLiou} for~$\cEL$.

First, consider a linear total differential operator $A$ whose arguments 
$\bphi(x,[w])={}^t\bigl(\phi_1,\ldots,\phi_r\bigr)$ are the variational
covectors for the infinite jet bundle over $\xi$.
  %belong to some horizontal module \cite{Opava}. % of sections.
Assume that the image of $A$ in the Lie algebra of evolutionary vector 
fields $\cEv_{\vph}$ %{A(\cdot)}
is closed w.r.t.\ the commutation:
$[\text{im}\,A, \text{im}\,A]\subseteq\text{im}\,A$. By the Leibnitz rule,
two sets of summands appear in the bracket of %evolutionary vector fields 
fields $A(\phi'),A(\phi'')$ that belong to the image of $A$: 
 %the %a Frobenius \\ operator $A$:
\[\bigl[A(\phi'),A(\phi'')\bigr]=A\bigl(\cEv_{A(\phi')}(\phi'')-
  \cEv_{A(\phi'')}(\phi')\bigr)+\bigl(
  \cEv_{A(\phi')}(A)(\phi'')-\cEv_{A(\phi'')}(A)(\phi')\bigr).\]
In the first summand we have used the permutability of
evolutionary derivations and total derivatives. The second
summand hits the image of $A$ by construction. %Theorem \ref{IspHO}.

The commutator $[\,,\,]\bigr|_{\mathrm{im}\,A}$ induces
   %a skew\/-\/symmetric \emph{Koszul bracket} 
a Lie algebra structure $[\,,\,]_A$
in the quotient $\Omega(\xi_\pi)$ %=\gf/\ker A$ 
of the domain of $A$ by its kernel:
\begin{subequations}\label{EqOplusB}
\begin{gather}
\bigl[A(\phi'),A(\phi'')\bigr]=A\bigl([\phi',\phi'']_A),\qquad
\phi',\phi''\in\Omega(\xi_\pi). %\equiv
   %\Gamma\bigl(\pi_\infty^*(\xi)\bigr)/\ker A.
   \label{EqOplusBBoth}\\
\intertext{This bracket, which is defined up to $\ker A$, equals}
[\phi',\phi'']_A=\cEv_{A(\phi')}(\phi'')-\cEv_{A(\phi'')}(\phi')+
  \{\!\{\phi',\phi''\}\!\}_A.\label{EqOplusBKoszul}
\end{gather}
\end{subequations}
It contains the two standard summands and
the skew\/-\/symmetric bilinear bracket $\{\!\{\,,\,\}\!\}_A$.

\begin{lemma}[\textup{\cite{Olver,Opava}}]\label{AllKnown}
The image of a Hamiltonian operator $\smash{\hat{A}}=
\bigl\|\sum\nolimits_\tau A^{\alpha\beta}_\tau(x,[w])\cdot D_\tau\bigr\|$
is closed w.r.t.\ the commutation.
The $k$-th component ($1\leq k\leq r$) of the 
arising bracket $\smash{\ib{\,}{\,}{\hat{A}}}$ %, which is
on the domain of~$\hat{A}$   %for a %the Hamiltonian operator 
is calculated by the formula  % \eqref{EqDogma},
\begin{equation}\label{EqDogma}
\ib{\phi'}{\phi''}{\hat{A}}^k=\sum_{|\sigma|\geq0}\sum_{i=1}^r (-1)^\sigma
 \Bigl(D_\sigma\circ\Bigl[\sum_{|\tau|\geq0}\sum_{j=1}^r D_\tau(\phi'_j)\cdot
 \frac{\dd A_\tau^{ij}}{\dd w^k_\sigma}\Bigr]\Bigr)
 \bigl(\phi''_i\bigr).
\end{equation}
The coefficients %of the Hamiltonian operator $_k$ and 
of the bilinear terms in the bracket $\ib{\,}{\,}{\hat{A}}$ are
differential functions of the variables~$w$.
\end{lemma}

Now we pass from the Hamiltonian operators~\eqref{Quattro} to the operators~$\square$
that have the same domain as~$\smash{\hat{A}_k}$ but take values in a different Lie algebra. Here is %we formulate 
our main result.
%Our main result is the following.

\begin{theor}\label{IspHO}%\label{FormulaForSquare}
Let the following conditions\footnote{The list can be not minimal such that
  it is easier to verify each requirement.}
be satisfied on an open dense subset of
the Euler\/--\/Lagrange system~$\cEL=\bigl\{\bu_{xy}=f(\bu;x,y
)\bigr\}$ 
of Liouville type \textup{(}see Proposition~\textup{\ref{NoetherSymTh}):}
\begin{itemize}
\item the constant symmetric real matrix~$\kappa$ in the kynetic term of the
Lagrangian density~$\cL$ be invertible\textup{;}
\item the linearization $\ell_f^{(u)}=\|\dd f^i/\dd u^j\|\equiv f'(\bu;x,y)$
of the right\/-\/hand side in~$\cEL$ be an invertible matrix\textup{;}
\item there be as many integrals $w^i\bigl(x,[\gm]\bigr)\in\ker D_y{\bigr|}_{\cEL}$ 
as there are unknowns~$u^j$\textup{;}
\item the integrals $w$ be minimal\textup{:} $\Phi\in\ker D_y{\bigr|}_{\cEL}$
implies $\Phi=\bigl(x,[w]\bigr)$\textup{;}
\item the integrals~$w$ be differential\/-\/functional 
independent,\footnote{\label{CompareLeftRight}The 
non\/-\/existence of a nontrivial $\Phi$, and hence of its nonzero
linearization $\ell_\Phi^{(\gm)}=\ell_\Phi^{(w)}\circ\ell_{w}^{(\gm)}$,
is equivalent to $\bigl(\nabla\circ\ell_w^{(\gm)}$ with 
$\nabla=\ell_{\bullet}^{(\gm)}\bigr)\Rightarrow\nabla=0$. This nondegeneracy
requirement is dual to~\eqref{NoKernelsIntersectNew}, see below.}
meaning
that $\Phi\bigl(x$, $\bigl[w[\gm]\bigr]\bigr)=0$ implies $\Phi\equiv0$\textup{;}
\item the $(r\times m)$-\/matrix $\Lambda=\bigl\|\dd w^i/\dd \gm^j_{d(i)}\bigr\|$,
where $d(i)\mathrel{{:}{=}}\ord_x w^i$ is the differential order of the $i$-\textup{th} integral $w^i[\gm]$, % w.r.t.~$\gm=-\tfrac{1}{2}\kappa u_x$, 
be invertible\textup{.}
%\item the linearization $\ell_{w}^{(\gm)}=\sum_i\Lambda_i[\gm]\cdot D_x^i$ 
%of the integrals~$w[\gm]$ w.r.t.\ their arguments be nondegenerate\textup{:}
%\[\bigcap_{i}\ker\Lambda_i = \{0\}.\]
\end{itemize}
Then the following statements hold\textup{:}
\begin{enumerate}
\item %\textup{(i)}\ %
%If the integrals $w$ are minimal for a nonlinear equation $\cEL$
%\textup{(}that is, $D_y(H)\doteq0$ implies $H=H\bigl(x,[w]\bigr)$\textup{)},
the image of the operator \eqref{Square} is closed w.r.t.\ the commutation
of symmetries
$\vph=\square\bigl(\boldsymbol{\phi}(x,[w])\bigr)\in\sym\cEL$\textup{;}
\item 
%if, additionally,
%\begin{itemize}
%\item the integrals $w=w[\gm]$ for~$\cEL$ are independent of~$x$,
%and only sections $\boldsymbol{\phi}=\boldsymbol{\phi}\bigl([w]\bigr)$
%are considered\textup{,}
%\end{itemize}
%then 
the bracket $\ib{\,}{\,}{\square}$ arising on the domain % $\gf$
of the operator $\square$ satisfies the equality
\begin{equation}\label{CalculateSokBrViaHam}
\ib{\phi'}{\phi''}{\square}=\ib{\phi'}{\phi''}{\hat{A}_k}
  %=\Bigl(\ell_{\phi',\hat{A}_k}^{(w)}\Bigr)^*(\phi'')
,\qquad \phi',\phi''\in\Omega(\xi_\pi); %\cosym\gA\subset\sym\cE_\varnothing;
\end{equation}
   %The bracket $\ib{\,}{\,}{\square}$ in the inverse image of
   %operator \eqref{Square}   %$\square$ %is equal to 
\item %\textup{(ii)}\ %
the coefficients %of the Hamiltonian operator $\hat{A}_k$ and 
of the bilinear terms in the bracket $\ib{\,}{\,}{\square}$ are
differential functions of the integrals~$w$.
\end{enumerate}
\end{theor}

\begin{rem}\label{RemOnMain}
The first assumption of the theorem implies that the system
$\cEL$ is determined, normal, and $\ell$-\/normal (see section~3%\ref{SecNonEL}
). We emphasize that $\cEL$~is the only system of equations imposed upon the
sections $u=s(x,y)\in\Gamma(\pi)$. % of the bundle~$\pi$.

%Unlike in~\cite{LeznovSmirnovShabat,ShabatYamilov}, our second assumption,
%the linear independence of the right\/-\/hand sides~$f$, becomes implicit 
%in~\cite{Shabat95} and seems to be necessary in~\cite{SokStar}.

Our second statement means that the ambiguity (up to $\ker\hat{A}_k$)
in the choice of a representative from the equivalence class 
$\ib{\phi'}{\phi''}{\hat{A}_k}$ in
the right\/-\/hand side of~\eqref{CalculateSokBrViaHam} amounts to the
choice of an element from $\ker\square\subseteq\ker\hat{A}_k$.
   %By showing that $\ker\hat{A}_k\subseteq\ker\square$, 
We prove the equality
of the kernels for all~$\boldsymbol{\phi}([w])\in\Omega(\xi_\pi)$.
This implies that commutation relations in~$\sym\cEL$,
which are determined by 
the Lie algebra structure~\eqref{EqOplusBKoszul} on $\Omega(\xi_\pi)$,
are obtained explicitly via~\eqref{EqDogma} for~$\smash{\hat{A}_k}$.

Being a corollary of Lemma~\ref{AllKnown} and the first two, 
our third statement is, at the same time, a special case
of Proposition~\ref{LCoeffInKernel} (see~below).
\end{rem}

\begin{proof}[Proof of Theorem~\textup{\ref{IspHO}}]
We notice first that %Noether 
symmetries~\eqref{SymForLiou} are independent of $u$ and of 
$u_y,u_{yy},\ldots$. Hence this is also true for the commutator
$\vph=[\vph',\vph'']\in\sym\cEL$ of two such symmetries
$\varphi'=\square\bigl(\boldsymbol{\phi}'(x,[w])\bigr)$ and 
$\varphi''=\square\bigl(\boldsymbol{\phi}''(x,[w])\bigr)$,
because the Lie bracket is a local bi\/-\/differential operator. 
   %(here, with $D_y$ not involved).

The factorization %decomposition
\eqref{Quattro} and Lemma~\ref{AllKnown} provide the diagram
\[
\begin{diagram}
{} && \dot{w}=\Phi\bigl(x,[w]\bigr) & \lto & 
    \bigl[{\phi'},{\phi''}\bigr]_{\hat{A}_k} \\
{} & \ruto^{\ell_w^{(\gm)}} & & \luto^{\hat{A}_k} & %\dto 
   \dAreEqual \\
\dot{\gm}=\psi\bigl(x,[\gm]\bigr) & & \circlearrowright & & \boldsymbol{\phi}=
 \boldsymbol{\phi}\bigl(x,[w]\bigr).\\
{} & \luto_{\hat{B}_1=\ell_{\gm}^{(u)}} & & \ldto_{\square} & {}\\
{} & & \dot{u}=\vph\bigl(x,[u_x]\bigr) & & {}
\end{diagram}
\]
%By construction, 
The commutator $\vph=[\vph',\vph'']$ determines the velocity
$\Phi\bigl(x,[w]\bigr)$ of the integrals that equals\footnote{The 
  $\ell$-\/normality of $\cEL$ implies that %the class
$\vph$~\emph{is} its symmetry whenever the velocity $\cEv_\vph(w)$ of the minimal integrals lies in $\ker D_y{\bigr|}_{\cEL}$, 
   see~\eqref{Preserve} in section~3.%\ref{SecNonEL}
} 
$\cEv_{[\vph',\vph'']}(w)=\bigl[\hat{A}_k(\phi'),\hat{A}_k(\phi'')\bigr]$.
Since the image of the Hamiltonian operator~$\hat{A}_k$ is closed under
commutation, we obtain the equivalence class 
$\boldsymbol{\phi}\bigl(x,[w]\bigr)=
\bigl[\phi',\phi''\bigr]_{\hat{A}_k}$ %\subset\Omega(\xi_\pi)$ 
of sections such that $\Phi=\hat{A}_k(\boldsymbol{\phi})=\cEv%^{(u)}
_{\square(\boldsymbol{\phi})}(w)$. By construction of~$\hat{A}_k$,
the commutator of $\vph'$ and $\vph''$ belongs to the set of symmetries
$\square(\boldsymbol{\phi})$. This proves the first statement of the theorem.

However, there may be many such $\vph=\square(\boldsymbol{\phi})$
%from this class
that induce the same velocity $\dot{w}=\Phi$. Since %occasionally may not hit
$\ker\square\subseteq\ker\hat{A}_k$, then, in principle, 
the equivalence class $[{\phi'},{\phi''}]_{\hat{A}_k}$ may contain
elements that do not belong to% the set
~$[{\phi'},{\phi''}]_{\square}$.

We claim that all representatives of the equivalence class
$[{\phi'},{\phi''}]_{\hat{A}_k}$ determine a unique symmetry $\vph$ 
of~$\cEL$. % under the map~$\square$. 
Therefore this 
$\vph=\square(\boldsymbol{\phi})$ is %which uniquely determines   inevitably
the commutator
$\bigl[\square(\boldsymbol{\phi}'),\square(\boldsymbol{\phi}'')\bigr]$ 
of the two symmetries, %~\eqref{SymForLiou}.
because the image of $\boldsymbol{\phi}$ under $\square$ must contain~it.
       %To verify the claim, 
It suffices to prove the uniqueness of the trivial solution~$\vph$ for the 
linear homogeneous equation $\ell_w^{(u)}(\vph)=\bigl(\ell_{w[\gm]}^{(\gm)}\circ\ell_{\gm}^{(u)}\bigr)(\vph)=0$.

There exist three 
ways to obtain the zero velocity $\Phi\bigl(x,[w]\bigr)
\equiv0$ of the integrals~$w\bigl[\gm[u]\bigr]$ along $\vph\in\sym\cEL$.
The first possibility is that the integrals be differentially dependent, 
which is excluded by the assumption of the theorem.
Second, the intermediate equation $\kappa\,D_x(\vph)=0$, $\det\kappa\neq0$,
may have nontrivial solutions only if some shifts $\vph=\text{const}$ are
symmetries of~$\cEL$. However, the determining equation 
$\bigl(D_xD_y-f'(\boldsymbol{u};x,y)\bigr)(\text{const})\doteq0$ on~$\cEL$ 
then exprimes the linear dependence between the differentials of its 
right\/-\/hand sides. This contradicts another initial assumption.

Therefore the proof of the second statement is reduced to the uniqueness
problem for the zero solution~$\psi%\bigl(%x, [\gm]\bigr)
=0$ of the linear homogeneous equation
$\ell_{w[\gm]}^{(\gm)}(\psi)=0$. %in total derivatives.
 %Notice that its restriction onto the jets of sections $\gm=r(x)$ must hold
 %for all sections at once.
Since $\cEL$ is the only system\footnote{The hyperbolic system~$\cEL$ is
\emph{formally integrable} \cite{ClassSym}: its infinite prolongation~$\cEL^\infty$ exists and there is an epimorphism $\cEL^\infty\to\cEL$.}
of differential relations that are imposed upon
the sections $u=s(x,y)\in\Gamma(\pi)$, the zero in the right\/-\/hand side of
$\ell_w^{(\gm)}(\psi)=0$ is achieved identically w.r.t.~$[\gm]$ (otherwise it would overdetermine% the system
~$\cEL$).
   %then either the equation upon~$\psi$ is 
   %the new differential constraint upon sections 
   %$u=s(x)$ of the bundle~$\pi$, which yields a contradiction, or

There remain two situations when a velocity $\psi$ of the momenta
$\gm=-\tfrac{1}{2}\kappa u_x$ makes $\dot{w}=\Phi$ zero.
 % (see also Remark~\ref{RemProofX} below). 
One reason is obvious: it is the use of the equation~$\cEL=\{\bu_{xy}=f\}$. 
Indeed, if the `time' along the %symmetry
flow $\dot{\gm}=\psi$ is the variable~$y$, \textit{i.e.}, 
$\gm_y=-\delta H_{\IL}/\delta u$,   %\tfrac{1}{2}\kappa f(\boldsymbol{u})$, 
then we have $\cEv_{\gm_y}^{(\gm)}(w)=
D_y\bigl(w[\gm]\bigr)\doteq0$ on~$\cEL$. But the presence of $\boldsymbol{u}$
in the list if arguments of~$f$ excludes such solutions~$\psi$ from further
consideration.

%From now on, we are forced to use the additional 
%assumption: $\dd w/\dd x=0$ and $\dd\boldsymbol{\phi}/\dd x=0$, whence
%it is sufficient to consider only the solutions $\psi=\psi[\gm]$ 
%(see Remark~\ref{RemProofX} that follows the proof).

Without loss of generality we assume that the integral~$w^r[\gm]$ has the highest
differential order: $d(r)\geq d(i)$ for all $i<r$. 
Let us calculate the velocities of the non\/-\/minimal integrals 
$(w')^i\mathrel{{:}{=}}D_x^{d(r)-d(i)}(w^i)$.
Using the permutability $\bigl[D_x,\cEv_\psi^{(\gm)}\bigr]=0$ of evolutionary
derivations with total derivatives, from the
identities $\cEv_\psi^{(\gm)}(w^1)=\dots=\cEv_\psi^{(\gm)}(w^r)=0$
we deduce that $\cEv_\psi^{(\gm)}(w')=0$.
It is readily seen that the linearization of the new integrals is of the form
$\ell_{w'}^{(\gm)}=\Lambda\cdot D_x^{d(r)}+O(d(r)-1)$, 
where the matrix~$\Lambda$ is invertible by our initial %additional 
assumption.
Multiplying the new equation $\ell_{w'}^{(\gm)}%\bigl
(\psi%[\gm]\bigr
)=0$
by $\Lambda^{-1}$, we obtain by induction 
that~$\psi\bigl(x,[\gm]\bigr)$ does not depend 
on~$[\gm]$, hence $\psi=\psi(x)$. 
   %%%
Consequently, the admissible sections $\varphi$ that solve the intermediate equation $\psi(x)=-\frac{1}{2}\kappa D_x(\varphi)$ also depend on $x$ only: 
$\varphi=\varphi(x)$. However, such sections, whenever nonzero, 
can not be\footnote{This is an immediate, point\/-\/by generalization of the fact (see above) that $\varphi=\text{const}\neq0$ is not a symmetry of~$\cEL$.}
symmetries of the hyperbolic system 
$u_{xy}-f(\boldsymbol{u};x,y)=0$ due to the nondegeneracy $\det f'(\boldsymbol{u})\neq 0$. In this notation, the ``symmetry'' $\varphi(x)$ must satisfy the determining equation $\bigl(D_x\circ D_y - f'(\boldsymbol{u})\bigr)\,\varphi(x) \doteq 0$ on~$\cEL$. The first summand vanishes because $D_y(x)\equiv0$. Thus we obtain $f'(u)\cdot\varphi(x)=0$, where the linearization matrix $f'(\boldsymbol{u})$ is invertible. Hence $\varphi(x)\equiv 0$. This completes the proof.
\end{proof}
 %The constant solution $\psi$ is redundant. Indeed,
 %if the free term~$\Lambda_0$ of the linearization
 %$\ell_w^{(\gm)}=\sum_{i=0}^{d(r)}\Lambda_i\cdot D_x^i$ 
 %is an invertible matrix,
 %then it imposes an excessive constraint upon sections of the bundle~$\pi$.
 %Even if $\det\Lambda_0\equiv0$, then the additional requirement 
 %$\vph=\vph[u_x]$
 %allows to quotient out all the solutions $\vph=h(x)$ of the equation
 %$-\tfrac{1}{2}\kappa D_x(\vph)=\text{const}$.
 %This concludes the proof, and we comment on it immediately.
  %But if there is a nontrivial solution $\psi\neq0$,
  %$\psi$ does depend on~$\gm\sim\kappa u_x$, 

 %Relaxing the $\pi$-\/verticality assumptions $\dd w/\dd x=0$ and
 %$\dd\boldsymbol{\phi}/\dd x=0$, 
\begin{rem}\label{RemProofX}
In the proof, we arrived at the linear ODE 
$\ell_w^{(\gm)}\bigl(\psi(x)\bigr)=0$ that holds simultaneously for all
sections $s\in\Gamma(\pi)$, although nonzero solutions~$\psi(x)$ do not
contribute to the symmetry algebra. 
   %%%%%%%
This is possible, first, if there is a total differential operator $\nabla$ such that $\ell_w^{(\gm)}\circ\nabla=0$.
(For instance, the identity $\left(\begin{smallmatrix}D_x & 1\\ 0 & 0\end{smallmatrix}\right)\binom{1}{-D_x}\bigl(h(x)\bigr)\equiv0$ holds for all~$h(x)$.) To avoid this, it is necessary 
\begin{itemize}
\item to require the \emph{nondegeneracy} of the linearization of the integrals:
\begin{equation}\label{NoKernelsIntersectNew}
\ell_{w[\gm]}^{(\gm)}\circ\nabla=0\ \Longrightarrow\ \nabla=0.
\end{equation}
In the adjoint form, 
$\nabla^*\circ\square=0\Rightarrow\nabla^*=0$, equation~\eqref{NoKernelsIntersectNew} exprimes the 
absence of linear differential relations\footnote{Thence 
the nondegeneracy~\eqref{NoKernelsIntersectNew} is analogous to the notion of $\ell$-\/normal differential equations in
the analysis of their formal integrability, see section~3.%\ref{SecNonEL}
}
between components of
symmetries~$\vph=\square\bigl(\boldsymbol{\phi}\bigl(x,[w]\bigr)\bigr)$.
\end{itemize}
This property is dual to the nondegeneracy $\nabla\circ\ell_w^{(\gm)}=0$ $\Rightarrow$ $\nabla=0$ that originates from the differential\/-\/functional independence $\Phi\bigl(x,[w]\bigr)=0$ $\Rightarrow$ $\Phi\equiv0$ via $\nabla=\ell_\Phi^w$, see footnote~\ref{CompareLeftRight} 
on p.~\pageref{CompareLeftRight}. 

Finally, let the section $s(x,y)\in\Sol\cEL$ be a solution of the 
Darboux\/-\/integrable Liouville\/-\/type system~$\cEL$. Taking the
restriction ${\EuScript L}^s=\ell_w^{(\gm)}{\bigr|}_{j_\infty(s)}$ of the linearization
operator onto the jet of~$s$, we obtain
the ordinary differential equation ${\EuScript L}^s\bigl(\psi(x)\bigr)=0$. 
For each~$s$, the linear space $\mathcal{O}(s)$
of its solutions is finite dimensional. (For example, its
dimension is equal to the sum of the exponents of a semi\/-\/simple complex
Lie algebra if $\cEL$ is the associated 2D~Toda chain.) Therefore 
\begin{itemize}
\item the requirement
\[
\bigcap\limits_{s\in\Sol\cEL\subset\Gamma(\pi)} \mathcal{O}(s)=\{0\},
\]
\end{itemize}
in combination with~\eqref{NoKernelsIntersectNew}, 
permits to eliminate the excessive freedom in the choice of solutions
$\psi(x)$ of the equation~$\smash{\ell_w^{(\gm)}}(\psi)=0$.
\end{rem}
  %complete the same proof in the full generality and thus extend 
  %the relation~\eqref{CalculateSokBrViaHam} onto the entire~$\Omega(\xi_\pi)$.
%%%%%%%%%%%%%%%%%%%%%
 %the matrix coefficients $\Lambda_i\bigl(x,[\gm]\bigr)$ of the expansion
 %$\ell_w^{(\gm)}=\sum_i\Lambda_i\cdot D_x^i$ have a nontrivial intersection
 %of the kernels and $\psi$~lies in it. 
 %That possibility %, which is conveniently checked in practice,
 %is eliminated by the last assumption of the theorem. 

Theorem~\ref{IspHO} %FormulaForSquare} 
is illustrated for semi\/-\/simple complex Lie algebras of rank two
in~\cite{Protaras}, where
the Hamiltonian operators $\smash{\hat{A}_1}$ and $\smash{\hat{A}_k}$
are constructed for the corresponding Drinfel'd\/--\/Sokolov 
hierarchies~\cite{DSViniti84} and the commutation relations in $\sym\cEL$ 
are calculated for the %corresponding 
2D~Toda chains $\bu_{xy}=\exp(K\bu)$. %\ c.f.\ Example~\ref{ExA2}.
   % and \cite{TMPhGallipoli}.

\begin{example}[The modified Kaup\/--\/Boussinesq equation]\label{ExIra}
Consider an Euler\/--\/Lagrange extension
of the scalar Liouville equation \cite{TMPhGardner},
\begin{equation}\label{EqIra}
A_{xy}=-\tfrac{1}{8}A\exp\bigl(-\tfrac{1}{4}B\bigr),\qquad
B_{xy}=\tfrac{1}{2}\exp\bigl(-\tfrac{1}{4}B\bigr).
\end{equation}
Denote the %Dirac \cite{Dirac}
momenta by
$%\[%\begin{equation}\label{m2KBConstraint}
a=\tfrac{1}{2}B_x$ and $b=\tfrac{1}{2}A_x$.
%\]%\end{equation}
The minimal integrals of system \eqref{EqIra} are
$%\[ %\begin{equation}\label{Integrals}
w_1=-\tfrac{1}{4}a^2-a_x$ and $w_2=ab+2b_x$
%\] %\end{equation}
such that $D_y(w_{i})\doteq0$ on \eqref{EqIra}, $i=1,2$.
Hence the operator
\[%\begin{equation}\label{SquareKBous}
\square=\Bigl(\ell_{w_1,w_2}^{(a,b)}\Bigr)^*=
\begin{pmatrix} -\tfrac{1}{4}B_x+D_x & \tfrac{1}{2}A_x\\
  0 & \tfrac{1}{2}B_x-2\,D_x \end{pmatrix}
\]%\end{equation}
determines (Noether, see \eqref{NoetherSym}) symmetries of \eqref{EqIra}.
The bracket $\ib{\,}{\,}{\square}$ induced in the inverse image of 
 %the Frobenius operator
$\square$ is
\[%\begin{equation}\label{SokKBous}  %[IraClassB.mw]/UU
\ib{\boldsymbol{\psi}}{\boldsymbol{\chi}}{\square} = \tfrac{1}{2}\cdot   
\begin{pmatrix}
 \psi^2_x\,\chi^1-\psi^1\,\chi^2_x +\psi^1_x\,\chi^2-\psi^2\,\chi^1_x\bigr) \mathstrut \\
 %\ib{\vec{\psi}}{\vec{\chi}}{\square}^2 = \tfrac{1}{2}\cdot\bigl(
 \psi^2_x\,\chi^2-\psi^2\,\chi^2_x\bigr)\mathstrut \end{pmatrix},
\]%\end{equation}
where $\boldsymbol{\psi}={}^t(\psi^1,\psi^2)$ and 
$\boldsymbol{\chi}={}^t(\chi^1,\chi^2)$;
we use upper indices for convenience.

Consider a symmetry of \eqref{EqIra},
\begin{equation}\label{pm2KB}
A_t=\tfrac{1}{2}A_xA_{xx}+\tfrac{1}{2}\left(\tfrac{1}{4}A_x^2-1\right)B_x,\qquad
B_t=-2A_{xxx}+\tfrac{1}{8}A_xB_x^2-\tfrac{1}{2}A_xB_{xx}.
\end{equation}
Let us %apply the twisting construction, see Remark \ref{RemTwistor}.
choose an equivalent %non\/-\/minimal
pair of integrals
$%\[ %\begin{equation}\label{m2KBKB}
u=w_2$, %ab+2b_x,
 %\qquad
$v=w_1+\tfrac{1}{4}w_2^2  %(u^2-a^2)-a_x
$. %\] %\end{equation}
  %It is remarkable that 
The evolution of $u$ and $v$ along \eqref{pm2KB}
equals %(see p. \pageref{KB})
\begin{equation}\label{KB}
u_t=uu_x+v_x,\qquad v_t=(uv)_x+ u_{xxx}.
\end{equation}
This is the Kaup\/--\/Boussinesq system, and \eqref{pm2KB} is %actually 
the potential twice\/-\/modified Kaup\/--\/Boussinesq 
equation, see \cite{PavlovConstCurv}.    %The above reasonings imply that
The right hand side of the integrable system \eqref{pm2KB} belongs to
the image of the adjoint linearization
$\widetilde{\square}=\bigl(\ell_{(u,v)}^{(a,b)}\bigr)^*$.
The %Frobenius 
operator $\widetilde{\square}$ factors
the \emph{third} Hamiltonian structure
$\hat{A}_3^{\text{KB}} = \widetilde{\square}^* \circ
\bigl(\ell_{(a,b)}^{(A,B)}\bigr)^* \circ \widetilde{\square}$
for \eqref{KB}; we have $k=3$ and
\[
\hat{A}_3^{\text{KB}}=
\begin{pmatrix}
u\,D_x+\tfrac{1}{2}u_x & D_x^3+(\tfrac{1}{4}u^2+v)\,D_x+\tfrac{1}{4}(u^2+2v)_x \\
D_x^3+(\tfrac{1}{4}u^2+v)\,D_x+\tfrac{1}{2}v_x &
   \tfrac{1}{2}(2u\,D_x^3+3u_x\,D_x^2+(3u_{xx}+2uv)D_x+u_{xxx}+(uv)_x)
\end{pmatrix}.
\]
By Theorem~\ref{IspHO}, %FormulaForSquare},
the bracket $\ib{\,}{\,}{\widetilde{\square}}$ is equal to
$\ib{\,}{\,}{\hat{A}_3^{\text{KB}}}$,
which is given by formula \eqref{EqDogma}. We obtain
\[%\begin{equation}\label{AlternativeNotation}
\ib{\boldsymbol{\psi}}{\boldsymbol{\chi}}{\widetilde{\square}} =
\ib{\boldsymbol{\psi}}{\boldsymbol{\chi}}{\hat{A}_3^{\text{KB}}} =
\begin{pmatrix}
\boldsymbol\psi\cdot\nabla_1(\boldsymbol\chi)
   -\nabla_1(\boldsymbol\psi)\cdot\boldsymbol\chi \\
\boldsymbol\psi\cdot\nabla_2(\boldsymbol\chi)
   -\nabla_2(\boldsymbol\psi)\cdot\boldsymbol\chi
\end{pmatrix},
\]%\end{equation}
where
%\[
$\nabla_1= - \tfrac{1}{2}\left(\begin{smallmatrix} D_x & 0 \\
   u\,D_x & D_x^3+v\,D_x \end{smallmatrix}\right)$
and %\qquad
$\nabla_2= - \tfrac{1}{2}\left(\begin{smallmatrix} 0 & D_x \\
   D_x & u\,D_x \end{smallmatrix}\right)$.
%\]
The operator $\hat{A}_1=\left(\begin{smallmatrix}0 & D_x\\
D_x & 0\end{smallmatrix}\right)$ is the first Hamiltonian structure
for \eqref{KB}; its inverse $A_1=\hat{A}_1^{-1}$ factors the second
Hamiltonian structure 
$B_2=\widetilde{\square}\circ A_1\circ\widetilde{\square}$
for \eqref{pm2KB}.%\marginpar{Detail.}
\end{example}

%Recursions.

\paragraph{3.\ Non\/-\/Lagrangian Liouville\/-\/type systems.}\label{SecNonEL}
%\noindent%
%The problem of construction of operators $\square$ that assign
%(possibly, not all) symmetries \eqref{SymForLiou}
%of non\/--\/Euler\/--\/Lagrange Liouville\/-\/type systems to 
%their integrals $w$, $\bar{w}$ is much less transparent. 
%A considerable progress has been achieved here
%in a recent paper \cite{SokStar} where it is shown
%that the existence of differential operators $\square$ 
%is based on the existence
%of $\bar{w}$, and respectively for $\bar{\square}$ and $w$.
Let $\cE=\{{F}=0\}$ be a Liouville\/-\/type system; 
now it may not be Euler\/--\/Lagrange. Let a column
$w\in\ker D_y\bigr|_{\cE}$ be %a section of $\xi$
composed by minimal integrals for $\cE$, thence $D_y(w)=\nabla({F})$
for some operator $\nabla$.
 %, see \eqref{Divergence} on p. \pageref{Divergence}. 
By construction of the Liouville\/-\/type systems $\cE$,
there are no differential relations (\emph{syzygies})
between the hyperbolic equations $\{F^i=0\}$ in them:
$\Delta({F})=0$ implies $\Delta=0$. 
For the same reason, the systems $\cE$ 
   %are independent from each other. Therefore $\cE$ is both normal and 
are $\ell$-\/normal \cite{ClassSym,Opava}: %, meaning that
$\Delta\circ\ell_F\doteq0$ on $\cE$ also requires $\Delta=0$. %respectively. 
Consequently, %By this argument, 
an evolutionary vector field $\cEv_\vph$ 
  %a section $\vph\in\varkappa(\pi)$ 
is a symmetry of a Liouville\/-\/type system $\cE$ if and only if it
preserves the integrals,
\begin{equation}\label{Preserve}
D_y\bigl(\cEv_\vph(w)\bigr)=\cEv_\vph(\nabla)(F)+\nabla\bigl(\cEv_\vph(F)\bigr)
\doteq\nabla\bigl(\ell_F(\vph)\bigr)\text{ on $\cE$.}
\end{equation}
  %This formula exprimes the content and the proof of Proposition 1 
  %in \cite{SokStar}.
Consider the operator equation
\[
D_y\circ\ell_w^{(u)}\doteq\nabla\circ\ell_F\text{ on $\cE$.}
\]
If, hypothetically, a total differential
operator $\square$ %\colon\gf\to\varkappa(\pi)$ in total derivatives 
such that %$\ell_w^{(u)}\circ\square\colon
 %\ker D_y\bigr|_{\cE}\to\ker D_y\bigr|_{\cE}$
\begin{equation}\label{OpEq} %\tag{\ref{OpEqSS}${}'$}
\ell_w^{(u)}\circ\square\in\CDiff\smash{\Bigl(\ker
  D_y\bigr|_{\cE}\to\ker D_y\bigr|_{\cE}\Bigr)}
\end{equation}
were constructed, %which thus resembles the right inverse of $\ell^{(u)}_w$, 
then it would assign symmetries $\vph=\square(\bphi)$
of the Liouville\/-\/type system $\cE$ to arbitrary $r$-\/tuples 
$\bphi\bigl(x,[w]\bigr)$ of the integrals, see \eqref{SymForLiou}.

The recent paper \cite{SokStar} contains an algorithm for construction of operator solutions $\square$ for the equation in total derivatives
\begin{equation}\label{OpEqSS}
\ell_w^{(u)}\circ\square=
  \bun_{m\times m}\cdot D_x^k\mod\CDiff_{<k}\smash{\Bigl(\ker
  D_y\bigr|_{\cE}\to\ker D_y\bigr|_{\cE}\Bigr)}.
\end{equation}
Most remarkably, the truncation from below for the se\-qu\-en\-ce of
coefficients of lower order derivatives in $\square$ is a consequence of the
presence of a complete set of the integrals $\bar{w}\in\ker D_x\bigr|_{\cE}$
for $\cE$.
However, the \emph{minimal} integrals $w$ must be `spoilt'
by differentiating w.r.t.\ $x$ a suitable number of times. Consequently,
instead of the \emph{Hamiltonian} operator $\hat{A}_k=\ell_w^{(u)}\circ\square$, see \eqref{Quattro}, one obtains the r.h.s.\ of \eqref{OpEqSS}.
Likewise, the images of operators constructed in \cite{SokStar} do not always span
the entire $x$\/-\/components of the Lie algebras $\sym\cE$, and the images
are generally not closed under the commutation. Moreover, the transformation rules
in the domains of $\square$ under reparametrizations $\tilde{w}[w]$ of the integrals remain unspecified %unclear 
for non\/-\/Lagrangian Liouville\/-\/type systems.

\begin{state}\label{LCoeffInKernel}
If the image of a solution $\square$ of the operator equation \eqref{OpEq}
 %with operators in total derivatives
for a Liouville\/-\/type system $\cE$ is closed under the commutation,
then all coefficients of the bracket $\ib{\,}{\,}{\square}$ on its domain,
see \eqref{EqOplusB}, belong to $\smash{\ker D_y\bigr|_{\cE}}$.
\end{state}

\begin{proof}
By assumption, we have that $\bigl(D_y\circ\ell_w^{(u)}\circ\square\bigr)
\bigl([\phi',\phi'']_\square\bigr)\doteq0$ for all
$\phi',\phi''(x,[w])$. %\in\gf$. 
This equals
\[%\begin{multline*}
0\doteq\bigl(D_y\circ\underline{\ell_w^{(u)}\circ\square}\bigr)
\Bigl(\cEv_{\square(\phi')}(\phi'')-\cEv_{\square(\phi'')}(\phi')+
\ib{\phi'}{\phi''}{\square}\Bigr)
\doteq\bigl(\ell_w^{(u)}\circ\square\bigr)\bigl(D_y
\bigl(\ib{\phi'}{\phi''}{\square}\bigr)\bigr),
\] %\end{multline*}
because the underlined composition satisfies \eqref{OpEq}. Clearly,
$D_y(\phi')$ and $D_y(\phi'')$ vanish on $\cE$ for arbitrary
$\phi',\phi''$. %\in\gf$. 
For the same reason, not only the whole bracket
$\ib{\phi'}{\phi''}{\square}$, but each particular coefficient
standing at the bilinear terms in it lies in $\ker D_y\bigr|_{\cE}$.
\end{proof}

\begin{example}\label{ExLiouE}
Consider the %one\/-\/
parametric extension of the scalar Liouville equation, % \eqref{ELiou},
\begin{equation}\label{LiouE}
\cE(\veps)=\bigl\{
u_{xy}=\exp(2u)\cdot\sqrt{1+4\veps^2u_x^2}\bigr\},\qquad
\veps\in\BBR.
\end{equation}
   %which is not Euler\/--\/Lagrange whenever $\veps\neq0$.
This equation is ambient w.r.t.\ the hierarchy of Gardner's
deformation of the potential modified KdV equation,
see \cite{TMPhGardner}.
The contraction $\cU=\cU(\veps,[u(\veps)])$
from \eqref{LiouE} to the
non\/-\/extended equation $\cU_{xy}=\exp(2\cU)$ is
$\cU=u+\tfrac{1}{2}\arcsinh(2\veps u_x)$; it
determines the third order integral
%\[
% w[\mu_\veps] = u_{xx}-u_x^2
%   + \frac{\veps u_{xxx}-2\veps u_xu_{xx}}{\sqrt{1+4\veps^2u_x^2}}
%   - \frac{\veps^2u_{xx}^2}{1+4\veps^2u_x^2}
%   - \frac{4\veps^3u_{xx}^2u_x}{(1+4\veps^2u_x^2)^{3/2}}
%\]
for \eqref{LiouE} using the integral $w=\cU_x^2-\cU_{xx}$ at $\veps=0$.
%see Example \ref{ExLiou}. 
However, the regularized (at $\veps=0$) %minimal 
integral of order two for \eqref{LiouE} is
\begin{equation}\label{MinIntLiouE}
w=\bigl(1-\sqrt{1+4\veps^2u_x^2}\bigr)\bigr/{2\veps^2}
   +u_{xx}\bigr/\sqrt{1+4\veps^2u_x^2}.
%\frac{1}{2\veps^2} - \frac{1+4\veps^2u_x^2-%2\veps^2u_{xx}}{2\veps^2\sqrt{1+4\veps^2u_x^2}};
\end{equation}
 %such that all other $x$-\/integrals for \eqref{LiouE} are differential
 %functions of \eqref{MinIntLiouE}.
The second integral for \eqref{LiouE} is
$\overline{w}=u_{yy}-u_y^2-\veps^2\cdot\exp(4u)\in\ker D_x\bigr|_{\cE(\veps)}$.
The operators $\bar{\square}=u_y+\tfrac{1}{2}D_y$ and
\begin{equation}\label{SquareE}%[SymComm.mw]
\square=
\tfrac{1}{2}(1+4 \veps^2 u_x^2-2 \veps^2 u_{xx})\cdot D_x+
u_x+4 \veps^2 u_x^3-2 \veps^2 u_{xxx}+
   \frac{12 \veps^4 u_x u_{xx}^2}{1+4 \veps^2 u_x^2}
\end{equation}
assign symmetries % \eqref{SymForLiou} 
$\vph=\square\bigl(\phi(x,[w])\bigr)$ and
$\overline{\vph}=\overline{\square}\bigl(\overline{\phi}(y,[\overline{w}])\bigr)$
of \eqref{LiouE} to its integrals.
%We emphasize that operators in the
%family \eqref{SquareE} assign \emph{higher} symmetries
%$\vph=\square\bigl(\phi(x)\bigr)$ of \eqref{LiouE} to functions on the
%base of the jet bundle whenever $\veps\neq0$, while the operator
%$\bar{\square}$ always determines point %classical
%symmetries $\bar{\vph}=\bar{\square}\bigl(\bar{\phi}(y)\bigr)$.
 %Hence the image of $\square$ does not cover the entire $x$-\/component
 %of those symmetry generators for \eqref{LiouE} which depend on free
 %functional parameters.

The images of both operators $\square$ and $\overline{\square}$ %are Frobenius.
are Lie subalgebras in $\sym\cE(\veps)$. %satisfy \eqref{EqDefFrob}.
The bracket $\ib{p}{q}{\overline{\square}}=p_yq-pq_y$ for $\overline{\square}$
is familiar \cite{SokolovUMN,TMPhGallipoli}. 
The bracket induced in the domain of $\square$ has the following form:
for any arguments $p,q$, %\in\Omega^1(\cE^\veps_\IL)$, 
we have
\begin{align*} %[SymComm.mw]
{}&\ib{p}{q}{\square} =
  \veps^2\cdot\bigl(p_{xx} q_x-p_x q_{xx}\bigr)
 -2\veps^2\cdot\bigl(p_{xxx} q-p q_{xxx}\bigr)\\
&\quad{}-12 \veps^4\cdot\bigl(8 \veps^2 u_x^3 u_{xx}-4 \veps^2 u_x^2 u_{xxx}+4 \veps^2 u_x u_{xx}^2
   +2 u_x u_{xx}-u_{xxx}\bigr)\\
&\qquad{}\times\bigl[1+8 \veps^2 u_x^2+16 \veps^4 u_x^4-2 \veps^2 u_{xx}
 -8 \veps^4 u_x^2 u_{xx}\bigr]^{-1}
 \cdot\bigl(p_{xx} q-p q_{xx}\bigr)\\
%\end{align*}
%
%\begin{align*}
{}&\quad{}+\bigl(\underline{1}+288 \veps^4 u_x^4-288 \veps^4 u_x^2 u_{xx}+28 \veps^2 u_x^2-16 \veps^2 u_{xx}-288 \veps^6 u_x u_{xx} u_{xxx}\\
&\qquad{}-96 \veps^6 u_{xx}^3+3072 \veps^{10} u_x^{10}+24 \veps^6 u_{xxx}^2+24 \veps^4 u_{4x}+1408 \veps^6 u_x^6
  +3328 \veps^8 u_x^8\\
&\qquad{}-768 \veps^{10} u_{4x} u_{xx} u_x^4 -384 \veps^8 u_{4x} u_x^2 u_{xx}
  -2304 \veps^8 u_x^3 u_{xx} u_{xxx}+384 \veps^8 u_{xx}^2 u_x u_{xxx}\\
&\qquad{}-4608 \veps^{10} u_x^5 u_{xx} u_{xxx}+16 \veps^4 u_{xx}^2-5632 \veps^8 u_x^6 u_{xx}
  -1920 \veps^6 u_{xx} u_x^4 +3328 \veps^8 u_x^4 u_{xx}^2\\
&\qquad{}+512 \veps^6 u_{xx}^2 u_x^2+384 \veps^{10} u_x^4 u_{xxx}^2
  -960 \veps^{10} u_{xx}^4 u_x^2-48 \veps^4 u_x u_{xxx}-3072 \veps^{10} u_x^7 u_{xxx}\\
&\qquad{}+3072 \veps^{10} u_{xx}^3 u_x^4
  -2304 \veps^8 u_x^5 u_{xxx}-576 \veps^6 u_x^3 u_{xxx}+288 \veps^6 u_{4x} u_x^2
  +384 \veps^8 u_x^2 u_{xx}^3\\
&\qquad{}+6144 \veps^{10} u_{xx}^2 u_x^6-6144 \veps^{10} u_{xx} u_x^8+1152 \veps^8 u_{4x} u_x^4
  +1536 \veps^{10} u_{4x} u_x^6
  +192 \veps^8 u_{xxx}^2 u_x^2\\
&\qquad{}+240 \veps^8 u_{xx}^4+1536 \veps^{10} u_{xx}^2 u_x^3 u_{xxx}
  -48 \veps^6 u_{4x} u_{xx}\bigr)\\
&\qquad{}\times\bigl[\underline{1}+96 \veps^4 u_x^4+256 \veps^6 u_x^6
  +256 \veps^8 u_x^8+4 \veps^4 u_{xx}^2-48 \veps^4 u_x^2 u_{xx}
  +32 \veps^6 u_{xx}^2 u_x^2\\
&\qquad\quad{}-4 \veps^2 u_{xx}-256 \veps^8 u_x^6 u_{xx}+64 \veps^8 u_x^4 u_{xx}^2
  -192 \veps^6 u_{xx} u_x^4+16 \veps^2 u_x^2\bigr]^{-1}\cdot\bigl(p_x
q-p q_x\bigr).
\end{align*}
%The surprisingly high
%differential orders of $\ib{\,}{\,}{\square}$ with respect to its
%arguments and coefficients is motivated by the presence of higher
%order derivatives of $u$ in \eqref{SquareE}.
The two underlined units correspond to the bracket $p_xq-pq_x$ on the
domain of the operator $\square=\cU_x+\tfrac{1}{2}D_x$ that %factors 
provides symmetries of the Liouville equation $\cU_{xy}=\exp(2\cU)$
at $\veps=0$. In agreement with Lemma \ref{LCoeffInKernel},
the non\/-\/constant coefficients of bilinear terms
$p_{xx}q-pq_{xx}$ and $p_xq-pq_x$ in $\ib{p}{q}{\square}$
belong to $\smash{\ker D_y\bigr|_{\cE(\veps)}}$.
It is remarkable that, since the entire construction 
(\ref{LiouE}--\ref{SquareE}) contains formal power series $u(\varepsilon)$
in $\varepsilon$,
so are these two rational functions: an attempt to express their dependence 
on $[w]$ leads to formal series with unbounded growth of the differential
orders of its coefficients.
\end{example}

\paragraph*{Discussion.} 
The matrix operators $\square=\bigl(\square^{i,j}$,\ $1\leq i\leq m$,\ %
$1\leq j\leq r\bigr)$ given by \eqref{Square} are generalizations of tensors
of type $(2,0)$ in the geometry of infinite jet bundles. We define the operators by using the two unrelated groups of differential reparametrizations for the coordinates in the domains and images, respectively.
Furthermore, the operators $\square$ for the Liouville\/-\/type systems $\cEL$
generalize the theory of Hamiltonian structures as follows: they map variational covectors for one equation (we recall that $\sym\cE_\varnothing\supset\gA$) to symmetries of the other system $\cEL$ 
(such that $\sym\cEL\supset\gB$).
  %Here it does not matter so much that $\square$ is not skew-adjoint,
  %c.f. \cite{SokolovUMN}.

Unlike in \cite{Demskoi2004,SokStar}, we do not attempt to solve 
equation \eqref{OpEq} upon~$\square$. 
On the contrary, we define the operators \eqref{Square}
by a geometric reasoning. Thence, first, we obtain the Hamiltonian operators 
$\smash{\hat{A}_k=\ell^{(u)}_w}\circ\square$ for the KdV\/-\/type hierarchies
on Euler\/--\/Lagrange systems of Liouville type 
\cite{DSViniti84,TMPhGallipoli} and, second,
we prove that the images of the operators~$\square$ are involutive. 
In other words, we describe a direct algorithm aimed at constructing %new
completely integrable equations.

Formulas \eqref{Square} and \eqref{Quattro} prescribe the differential order
of $\smash{\hat{A}_k}$. %and explain its structure.
Estimates for the orders of the integrals $w$ for the 2D~Toda chains associated with semi\/-\/simple complex Lie algebras $\gothg$ 
are known from \cite{ShabatYamilov}, see Example \ref{ExTodaAreLiou},
and were reformulated %were claimed or
in \cite{%LeznovSmirnovShabat,
GurievaZhiber,%Shabat95,
Protaras}. %in various formulations, . 
The \emph{upper}
bound, that the numbers $\ord_x w^i-1$ are not greater than the exponents
for $\gothg$, is proved by verifying (via Schur polynomials) Serre's relations
$(\ad Y_i)^{-K^i_{\,j}+1}(Y_j)=0$, $i\neq j$, for the generators
\[Y_i=\sum\nolimits_{k\geq0}\exp\Bigl(-\sum\nolimits_{j=1}^m K^i_{\,j}u^j\Bigr)\cdot D_x^k\Bigl(\exp
\Bigl(\sum\nolimits_{j'=1}^m K^i_{\,j'}u^{j'}\Bigr)\Bigr)\cdot \dd/\dd u^i_{k+1}\]
of the characteristic Lie algebras (see \cite{LeznovSmirnovShabat,ShabatYamilov,Shabat95} and also \cite{Protaras}),
and by using Frobenius theorem.
The fact that the vector fields $Y_i$ %are precisely 
coincide %\cite{ShabatYamilov} 
with the Chevalley 
generators $\mathfrak{f}_i$ of the semi\/-\/simple Lie al\-ge\-bra $\gothg$
is important here.
The same estimate from \emph{below} %can be proved by showing that there are
follows from the absence of relations other than Serre's for the
generators $Y_i$. %of the characteristic Lie algebras. %Equivalently, 
This was established in \cite[p.21]{ShabatYamilov}
for the root systems $\mathsf{A}_n$ and $\mathsf{D}_n$ %\textit{e.g.}, 
by listing the linear independent 
nonzero iterated commutators.

{%\small
\paragraph*{Acknowledgements.}
The authors thank B.\,A.\,Dub\-ro\-vin, E.\,V.\,Ferapontov,
I.\,S.\,Kra\-sil'\-sh\-chik, %M.\,A.\,Ne\-ste\-ren\-ko,
P.\,J.\,Ol\-ver, and V.\,V.\,Sokolov
for %useful 
discussions and remarks. %and constructive criticisms.
The authors acknowledge helpful advice of the referee.
A.\,K. is grateful to the organizing committee of the International workshop
`Nonlinear Physics: Theory and Experiment V' for support and warm hospitality.
This work has been partially supported by the European Union through
the FP6 Marie Curie RTN \emph{ENIGMA} (Contract
no.\,MRTN-CT-2004-5652), the European Science Foundation Program
{MISGAM}, and by NWO grants B61--609 and VENI 639.031.623.
A part of this research was done while A.\,K.\ was visiting
at the $\smash{\text{IH\'ES}}$, SISSA, and~CRM (Mont\-r\'eal), whose financial support is gratefully acknowledged.

}

\end{document}